\documentclass[a4paper,UKenglish,autoref,cleveref,thm-restate,final,numberwithinsect]{lipics-v2021}

\pdfoutput=1 %
\hideLIPIcs  %

\title{Linear Loop Synthesis for Quadratic Invariants}

\author{S. Hitarth}{Hong Kong University of Science and Technology, Hong Kong}{hsinghab@connect.ust.hk}{0000-0001-7419-3560}{}%

\author{George Kenison}{Liverpool John Moores University, United Kingdom}{g.j.kenison@ljmu.ac.uk}{0000-0002-7661-7061}{}

\author{Laura Kov\'{a}cs}{TU Wien, Austria}{laura.kovacs@tuwien.ac.at}{0000-0002-8299-2714}{}

\author{Anton Varonka}{TU Wien, Austria}{anton.varonka@tuwien.ac.at}{0000-0001-5758-0657}{}

\authorrunning{S. Hitarth, G. Kenison, L. Kov\'{a}cs, and A. Varonka} %

\Copyright{S. Hitarth, George Kenison, Laura Kov\'{a}cs, and Anton Varonka} %



\category{} %

\funding{\begin{wrapfigure}{l}{1cm}
		\includegraphics[width=1cm]{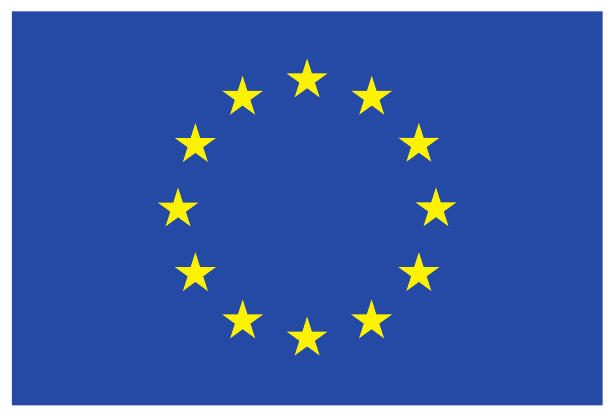}
	\end{wrapfigure}
	The project leading to this publication has received funding from the European Union's Horizon 2020 research and innovation programme under grant agreement No 101034440.
	The project was also partially supported by the ERC consolidator grant ARTIST
101002685 and by the Vienna Science and Technology Fund (WWTF) [10.47379/ICT19018].}

\nolinenumbers %

\EventEditors{Olaf Beyersdorff, Mamadou Moustapha Kant\'{e}, Orna Kupferman, and Daniel Lokshtanov}
\EventNoEds{4}
\EventLongTitle{41st International Symposium on Theoretical Aspects of Computer Science (STACS 2024)}
\EventShortTitle{STACS 2024}
\EventAcronym{STACS}
\EventYear{2024}
\EventDate{March 12--14, 2024}
\EventLocation{Clermont-Ferrand, France}
\EventLogo{}
\SeriesVolume{289}
\ArticleNo{16}

\usepackage[usenames,dvipsnames,table]{xcolor}
\definecolor{red}{RGB}{116,0,1}
\definecolor{blue}{RGB}{14,26,164}

\usepackage{amsmath, amsthm, amssymb, mathrsfs, calc, bm}
\usepackage{mathtools}
\usepackage{url}
\usepackage{mleftright}
\usepackage{microtype}
\usepackage[small,basic]{complexity}
\usepackage{aliascnt}
\usepackage{tikz}
\usepackage{booktabs}
\usetikzlibrary{calc,decorations.markings,positioning}
\usepackage{pgfplots}
\usepgfplotslibrary{polar}\usepgflibrary{shapes.geometric}\usepgflibrary{plotmarks}

\usepackage{graphicx}
\usepackage{wrapfig}

\usepackage[T1]{fontenc}
\usepackage[utf8]{inputenc} \DeclareUnicodeCharacter{2212}{-}

\usepackage{algpseudocode}
\usepackage{algorithm}

\NewDocumentCommand{\LeftComment}{s m}{%
  \Statex \IfBooleanF{#1}{\hspace*{\ALG@thistlm}}\textcolor{gray}{\(\triangleright\) #2}}
\makeatother

\algnewcommand\algorithmicswitch{\textbf{switch}}
\algnewcommand{\ConExp}[3]{$\;\left( #1 \right)\ ?\ #2 : #3 $}
\algnewcommand{\Break}{\textbf{break}}

\algdef{SE}[ITEM]{Item}{EndItem}[1]{\textbullet #1}{}%
\algdef{SE}[SWITCH]{Switch}{EndSwitch}[1]{\algorithmicswitch\ #1}{\algorithmicend\ \algorithmicswitch}%
\algtext*{EndItem}%

\usepackage{caption}
\usepackage{subcaption}

\usepackage{hyperref}

\newaliascnt{prob}{theorem}
\newtheorem{problem}[prob]{Problem}
\aliascntresetthe{prob}

\newtheorem*{que}{Question}

\newcommand{\appref}[1]{\hyperref[#1]{Appendix~\ref*{#1}}}

\renewcommand{\top}{\ensuremath\mathsf{T}}
\renewcommand{\epsilon}{\ensuremath\varepsilon}

\renewcommand{\Lambda}{\ensuremath\mu}

\renewcommand{\R}{\mathbb{R}}
\newcommand{\Z}{\mathbb{Z}}

\newcommand{\Q}{\mathbb{Q}}

\newcommand{\ZZ}{\mathbb{Z}}
\newcommand{\QQ}{\mathbb{Q}}

\newcommand{\KK}{\mathbb{K}}

\newcommand{\diag}[1]{\ensuremath{\operatorname{diag}(#1)}}
\DeclareMathOperator{\GL}{GL}

\newcommand*\wc{{ }\cdot{ }}

\providecommand*{\eu}%
{\ensuremath{\mathrm{e}}}
\providecommand*{\iu}%
{\ensuremath{\mathrm{i}}}

\newcommand{\seq}[3][{}]{\langle #2 \rangle_{#3}^{#1}}

\newcommand{\loopie}{\mathcal{L}}

\begin{document}
\maketitle

\begin{abstract}
Invariants are key to formal loop verification as they capture loop properties that are
        valid before and after each loop iteration. Yet, 
	generating invariants is a notorious task already for
        syntactically restricted classes of loops. 
        Rather than generating invariants for given loops, in this paper we
        synthesise loops that exhibit a
        predefined behaviour given by an invariant.
	From the perspective of formal loop verification, 
	the synthesised loops are thus correct by design and no longer
        need to be verified. 
	
	To overcome the hardness of reasoning with arbitrarily strong
        invariants, in this paper we construct simple (non-nested)
        \textbf{while} loops with linear updates that exhibit polynomial equality invariants.
Rather than solving arbitrary
        polynomial equations, we consider loop properties defined by a single quadratic invariant in
        any number of variables. We present a procedure that, given a quadratic equation, decides whether a loop with affine updates satisfying this equation exists. Furthermore, if the answer is positive, the procedure synthesises a loop and ensures its variables achieve infinitely many different values.
	
\end{abstract}
\section{Introduction} \label{sec:introduction}
Linear loops, in their simplicity, constitute a convenient and yet expressive model. 
From an algebraic point of view,  a linear loop corresponds to a system of recurrence relations; 
solutions of such systems form a robust class in algorithmic
combinatorics and algebraic number theory~\cite{everest2003recurrence,DBLP:series/tmsc/KauersP11}. %
Linear loops are particularly common in control and digital signal processing software~\cite{jeannet2014accel}.
Note also that the problem of studying the functional behaviour of affine loops
(loops with update polynomials of degree~1)
can be reduced to that of studying linear loops~\cite{ouaknine2015survey}.
Moreover, linear loops can be used to overapproximate the behaviour of more expressive numerical programs, including those with unrestricted control flow and recursive procedures~\cite{kincaid2019closed}.

\subparagraph*{Loop Invariants.}
While variable updates of linear loops are restricted to linear
assignments, it is quite common that linear loops exhibit intricate polynomial
properties in the form of polynomial invariants. %
Non-linear polynomial invariant assertions might come in handy for the verification of safety properties;
by approximating the program's behaviour more accurately,
they admit fewer false positives. 
That is, a program verifier using polynomial loop invariants infers less frequently that a true assertion can
be violated~\cite{chatterjee2020inv}.

\subparagraph{Loop Synthesis.}
Generating invariants, in particular polynomial invariants, is a
notorious task,  shown to be undecidable for loops with arbitrary polynomial
arithmetic~\cite{Ouaknine23}. 
Rather than generating invariants for
loops, \emph{in this paper we work in the reverse direction: generating loops
from invariants}.
Thus we ensure that the constructed loops exhibit intended
invariant properties and are thus correct by design. 
Loop synthesis therefore provides
an alternative approach for proving program correctness.
If intermediate assertions of an involved program are written in terms of polynomial equalities,
automated loop synthesis can provide a code fragment satisfying that assertion,
while being correct by construction with respect to the specification.

To overcome hardness of polynomial reasoning and solving arbitrary
polynomial equations, we restrict our
attention to linear loops, and provide a decision procedure
for computing linear loops from (quadratic) polynomial invariants (\autoref{alg:lin}). 

Linear loop synthesis showcases how a simple model (a linear loop) 
can express complicated behaviours (quadratic invariants), as also
witnessed in sampling algorithms of 
real algebraic geometry~\cite{basu2006algorithm,dufresne2018sampling}.
A non-trivial linear loop for a polynomial invariant allows to
sample infinitely many points from the algebraic variety defined by the
polynomial.
Moreover, the computational cost to generate a new sample point only involves a matrix-vector multiplication.
We give further comment on why we do not accept trivial loops in the synthesis process in \cref{rem:SyntSolv}.

Thus
the result of a loop synthesis process for a polynomial equation (invariant) is an infinite family of solutions defined by recurrence relations.
This family is parameterised by~$n$, the number of loop iterations:
$n$th terms of the synthesised recurrence sequences yield a solution of the polynomial equation.
Whether the solution set of an equation admits a parameterisation of a certain kind
is, in general, an open problem~\cite{skolem1938dioph,vaserstein2010param}.

\subparagraph*{Our Contributions.}
The main contributions of this work are as follows:%
\begin{enumerate}%
	\item We present a procedure that, given a quadratic equation $P(x_1, \dots, x_d) = 0$
with an arbitrary number of variables and rational coefficients,
generates an affine loop such that $P = 0$ is invariant under its execution; i.e.,
the equality holds after any number of loop iterations.
If such a loop does not exist, the procedure returns a negative answer.

The values of the loop variables are rational.
Moreover, the state spaces of the loops synthesised by this procedure are infinite and, notably, the same valuation of loop variables is never reached twice.
The correctness of this procedure is established in \autoref{loop-for-qe}. 
	\item If the equation $Q(x_1, \dots, x_d) = c$ under consideration is such that $Q$ is a quadratic form,
we present a stronger result: a procedure (\autoref{alg:lin}) 
that generates a \emph{linear} loop with~$d$ variables satisfying the invariant equation.
\end{enumerate}

\subparagraph*{Paper Outline.}
\cref{sec:prelims} introduces relevant preliminary material.
We defer the discussion of \emph{polynomial equation solving}, a key element of loop synthesis, to~\cref{sec:eqsolv}.
Then, in \cref{sec:qf}, we provide a method to synthesise \emph{linear loops}
for invariants, where the invariants restricted to be equations with \emph{quadratic forms}.
We extend these  results in~\cref{sec:qe} and present 
a procedure that synthesises \emph{affine loops}, and hence also
linear loops, for invariants that
are 
\emph{arbitrary quadratic} equations. 
We discuss aspects of our approach and propose further
directions in~\cref{sec:conclusion}, in relation to known results.

In this extended version of our STACS 2024 paper, we include two appendices.
In Appendix~A , we summarise the procedure for finding isotropic solutions to quadratic forms (which we employ in our synthesis procedure).
The abstract arithmetic techniques contained therein are beyond the scope of this short paper and detail the contributions of many sources \cite{ireland1990classical, cremona2003conics,simon2005solving,simon2005quadratic,meyer1884mathematische,castel2013solving}.
In Appendix~B, we summarise the synthesis procedure underlying \autoref{loop-for-qe}.
\section{Preliminaries} \label{sec:prelims}
\subsection{Linear and quadratic forms}
\begin{definition}[Quadratic form]
	A \emph{$d$-ary quadratic form} over the field~$\KK$ is a homogeneous polynomial of degree~2 with~$d$ variables: $$Q(x_1, \dots, x_d) = \sum_{i\leq j} c_{ij} x_ix_j,$$
	where $c_{ij} \in \KK$. It is convenient to associate a quadratic form \(Q\) with the symmetric matrix: $$A_Q := \begin{pmatrix}
		c_{11} & \frac{1}{2}c_{12} & \dots & \frac{1}{2}c_{1d} \\
		\frac{1}{2}c_{12} & c_{22} & \dots & \frac{1}{2}c_{2d} \\
		\vdots  & \vdots  & \ddots & \vdots  \\
		\frac{1}{2}c_{1d} & \frac{1}{2}c_{2d} & \dots & c_{dd} \\
	\end{pmatrix}.$$%
\end{definition}
We  note that since $A_Q$ is symmetric, its eigenvalues are all real-valued. Further,  $Q(\bm{x}) = \bm{x}^\top  A_Q \bm{x}$ for a vector~$\bm{x} = (x_1, \dots, x_d)$ of variables.

We consider quadratic forms over the field~$\QQ$ of rational numbers by default. Therefore, a quadratic form has a rational quadratic matrix associated with it.
	
A quadratic form~$Q$ is \emph{non-degenerate} if its matrix~$A_Q$ is not singular; that is, $\det{A_Q} \neq 0$. 
A quadratic form $Q$ over $\QQ$ \emph{represents} the value~$a \in \QQ$ if there exists a vector~$\bm{x} \in \QQ^d$ such that $Q(\bm{x}) = a$. 
A quadratic form $Q$ over $\QQ$ is called \emph{isotropic} if it represents~$0$ non-trivially; i.e., there exists a non-zero vector~$\bm{x} \in \QQ^d$ with $Q(\bm{x}) = 0$. The vector itself is then called \emph{isotropic}. If no isotropic vector exists, the form is \emph{anisotropic}. 
A quadratic form~$Q$ is called \emph{positive} (resp.~\emph{negative}) \emph{definite} if $Q(\bm{x}) > 0$ (resp.~$Q(\bm{x}) < 0$) for all $\bm{x} \neq \bm{0}$.
Note that definite forms are necessarily anisotropic.
\begin{definition}\label{def-eq-qf}
	Let $Q_1$ and $Q_2$ be $d$-ary quadratic forms. 
	The forms $Q_1$ and $Q_2$ are \emph{equivalent}, denoted by $Q_1 \sim Q_2$, if there exists $\sigma \in \GL_d(\QQ)$ such that $Q_2 (\bm{x}) = Q_1 (\sigma\cdot\bm{x})$.
\end{definition}

From the preceding definition, there exists an (invertible) linear change of variables over $\QQ$ under which representations by~$Q_2$ are mapped to the representations by~$Q_1$.
It is clear that two equivalent quadratic forms represent the same values.
In terms of matrices, we have $(\sigma \bm{x})^\top A_{Q_1}\sigma \bm{x} = \bm{x}^\top A_{Q_2}\bm{x}$, and hence
\(
	A_{Q_2} = \sigma^\top A_{Q_1}\sigma\).

	\begin{definition}[Linear form]
		A \emph{linear form} in $d$ variables over the field~$\QQ$ is a homogeneous polynomial $L(x_1, \dots, x_d) = \sum_{i=1}^db_ix_i$ of degree~1, where $b_1, \dots, b_d \in \QQ$.
	\end{definition}
Note that each linear form admits a vector interpretation: $L(\bm{x}) = \bm{b}^\top \bm{x}$, where $\bm{b} = (
	b_1, \ldots, b_d)^\top \in \QQ^d$ is a non-zero vector of the linear form.

\subsection{Loops and Loop Synthesis}
Linear loops are a class of single-path loops whose update assignments are determined by a homogeneous system of linear equations in the program variables.

\begin{definition}[Linear loop]
	A \emph{linear loop} $\langle M, \bm{s} \rangle$ is a loop program of the form
	\begin{equation*}
		\bm{x} \leftarrow \bm{s}; \enspace \textbf{while}\ \star\ \textbf{do} \enspace \bm{x}\leftarrow M \bm{x},
	\end{equation*}
\end{definition} 
where \(\bm{x}\) is a \(d\)-dimensional column vector of program variables, 
\(\bm{s}\) is an \emph{initial} $d$-dimensional \emph{vector},
and \(M\) is a \(d\times d\) \emph{update matrix}. %
For the procedures, which we introduce here, to be effective, we assume that the entries of~$M$ and~$\bm{s}$ are rational.

We employ the notation~\(\star\), instead of using
\texttt{true} as loop guard, as our focus is on loop synthesis rather
than proving loop termination.

\begin{definition}[Affine loop]
		An \emph{affine loop}~$\langle M, \bm{s}, \bm{t}\rangle$ is a loop program of the form
		\begin{equation*}
			\bm{x} \leftarrow \bm{s};\ \textbf{while}\ \star\ \textbf{do}\ \bm{x}\leftarrow M \bm{x} + \bm{t},
		\end{equation*}
where, in addition to the previous definition, $\bm{t} \in \QQ^d$ is a translation vector.
\end{definition}

\begin{remark}[Linear and Affine Loops]\label{rem:LinAffLoop}
  A standard observation permits the simulation of affine loops by linear ones at a cost of one additional variable constantly set to 1.
An \emph{augmented matrix} of an affine loop with~$d$ variables is a matrix \(M'\in\Q^{(d+1)\times (d+1)}\) of the form
	\[M':=\left(\begin{array}{c|c} 1 & 0_{1,d} \\ \hline \bm{t}  & M	\end{array}\right).\]
	It follows that a linear loop $\langle M', (1, \bm{s})^\top
        \rangle$ simulates the affine loop in its last $d$ variables.
\end{remark}
	
A linear (or affine) loop with variables~$\bm{x} = (x_1, \dots, x_d)$ generates~$d$ sequences of numbers.
For each loop variable $x_j$, let $\seq[\infty]{x_j(n)}{n=0}\subseteq \Q$ denote the sequence whose
$n$th term is given by the value of $x_j$ after the $n$th loop iteration.
Similarly, define the sequence of vectors~$\seq{\bm{x}(n)}{n}\subseteq \QQ^d$.
For a given loop, we refer to its reachable set of states in \(\QQ^d\) as the loop's \emph{orbit}.
	A loop with variables $x_1,\ldots, x_d$ is \emph{non-trivial} if the orbit
	\begin{equation*}
		\mathcal{O}_{\bm{x}} := %
		\{ \left( x_1(n) , \ldots, x_d(n)\right) : n\geq 0 \} \subseteq \QQ^d
	\end{equation*}
	is infinite. A \emph{polynomial invariant} of a loop is a polynomial~$P \in \QQ[\bm{x}]$ such that \[P(x_1(n), \dots, x_d(n)) = 0\] holds for all~$n\geq 0$. 
	
\begin{problem}[Loop Synthesis]\label{prob:synthesis}
	Given a polynomial invariant~$P \in \QQ[x_1, \dots, x_d]$, 
	find a non-trivial linear (affine) loop 
	with vector sequence~$\seq{\bm{x}(n)}{n}$ such that
	\[P(x_1(n), \dots, x_d(n)) = 0\]
	holds for any $n \geq 0$.
\end{problem}

We emphasise that, unless stated otherwise, the objective of the
loop synthesis process from \cref{prob:synthesis} is to find a loop with the same number of variables~$d$ as in the input invariant.
That is, $\seq{\bm{x}(n)}{n} = \left(\seq{x_1(n)}{n}, \dots, \seq{x_d(n)}{n}\right)$

Note that~$P = 0$ in \cref{prob:synthesis} does not need to be an \emph{inductive invariant} for the synthesised loop:
We do not require the matrix~$M$ to preserve the equality for all vectors~$\bm{x}$.
There might still exist a vector $\bm{s'}$ such that $P(\bm{s'}) = 0$ but $P(M \cdot \bm{s'}) \neq 0$.
Observe that the search space only expands when we allow non-inductive invariants, thus making our loop synthesis procedures more general.

In summary, the search for an update matrix~$M$ 
(or the augmented matrix~$M'$ in the affine loop version of~\cref{prob:synthesis}),
is integrally linked to the search of~$\bm{s}$, a solution of the
polynomial~$P = 0$.

\begin{remark}[Loop Synthesis and Polynomial Equation
  Solving]\label{rem:SyntSolv}
 We note that \cref{prob:synthesis}, Loop Synthesis, relies on, but it is not equivalent
to, solving polynomial equations. 
Indeed, we focus on \emph{non-trivial} loops in
\cref{prob:synthesis}. 
Allowing loops with finite orbits would mean that a loop with an identity matrix update~$I_d$ is accepted as a solution:
\[\label{trivsol}
\bm{x}\leftarrow \bm{s}; \enspace \textbf{while}\ \star\ \textbf{do} \enspace \bm{x}\leftarrow I_d \cdot \bm{x}.
\]
Then, the loop synthesis problem would be equivalent to the problem of finding a rational solution of a polynomial equation~$P=0$ (see~\cref{prob:dioph}).
The problem, as we define it in~\cref{prob:synthesis}, neglects loops that satisfy a desired invariant but reach the same valuation of variables twice.
Due to this, the  \cref{prob:synthesis} of loop synthesis is 
different from the~\cref{prob:dioph} of solving polynomial equations.
\end{remark}
\section{Solving Quadratic Equations}\label{sec:eqsolv}
As showcased in \cref{prob:synthesis} and discussed in \cref{rem:SyntSolv}, loop synthesis for a polynomial invariant~$P=0$ is closely
related to the problem of solving a polynomial equation~$P=0$.
\begin{problem}[Solving Polynomial Equations]\label{prob:dioph}
	Given a polynomial $P \in \QQ[x_1, \dots, x_d]$, decide whether there exists a rational solution $(s_1, \dots, s_d) \in \QQ^d$ to the equation $P(x_1, \dots, x_d) = 0$.
\end{problem}
We emphasise that determining whether a given polynomial equation has a rational solution, is a fundamental open problem in number theory~\cite{shlapentokh2011h10}, see also~\cref{sec:related}.

Clearly, this poses challenges to our investigations of loops satisfying arbitrary polynomial invariants.
In light of this, it is natural to restrict \cref{prob:synthesis} to loop invariants given by \emph{quadratic} equations.
Given a single equation~$P(\bm{x})= 0$ of degree~2, the challenge from now on is to find a rational solution~$\bm{s}$ and an update matrix~$M$ such that iterative application of $M$ to~$\bm{s}$ of the equation does not violate the invariant: $P(M^n\bm{s}) = 0$ for all~$n \geq 0$.

In this section, we recall well-known methods for solving quadratic equations.
In the sequel, we will employ said methods in the novel setting of loop synthesis for quadratic polynomial
invariants (\cref{sec:qf,sec:qe}).

\begin{problem}[Solving Quadratic Equations]\label{prob:quadr}
	Given a quadratic equation in~$d$ variables with rational coefficients, decide whether it has rational solutions. If it does, generate one of the solutions. 
\end{problem}
\subsection{Solutions of Quadratic Equations in Two Variables}
We first prove two lemmas that discuss the solutions of binary quadratic forms in preparation for~\cref{sec:qf}.
\begin{lemma}\label{bin1sol}
	
	For all $a,b \in \QQ\setminus\{0\}$,
	Pell's equation \(x^2+\frac{b}{a}y^2=1\)
	has a rational solution~$(\alpha, \beta)$ with $\alpha \not\in \{\pm 1, \pm \frac{1}{2}, 0\}$ and $\beta \neq 0$.
\end{lemma}
\begin{proof}
	So long as $a \neq -b$, it is easy to see that $\mleft(\frac{b-a}{a+b}, \frac{2a}{a+b}\mright)$ is a rational solution to Pell's equation.
	Recall that $a \neq 0$, hence $\beta \neq 0$ and $\alpha \neq \pm 1$.
	However, the generic solution might have $\alpha = 0$ or $|\alpha| = \frac{1}{2}$.
	We thus explicitly pick alternative solutions for the cases when it occurs:
	(i) $x^2 + y^2 = 1$ has another rational point, e.g., $(\frac{3}{5}, \frac{4}{5})$;
	(ii) $x^2+3y^2 = 1$ has a rational point $(-\frac{11}{13}, \frac{4}{13})$;
	(iii) $x^2+\frac{1}{3}y^2 = 1$ has a rational point $(\frac{1}{7},\frac{12}{7})$. 
	
	Finally, if $a = -b$, we can take a rational point $(\frac{5}{3},\frac{4}{3})$ on the hyperbola $x^2-y^2=1$.	
\end{proof}

\begin{lemma}\label{bin-solutions}
	An equation $ax^2+by^2=c$ with $a, b \in \QQ \setminus{0}$ has either no rational solutions different from \((0,0)\), or infinitely many rational solutions different from $(0,0)$. 
\end{lemma}

\begin{proof}
	Define \( R:= \left(\begin{smallmatrix} 
				\alpha & -\frac{b}{a}\beta\\
		\beta & \alpha \end{smallmatrix}\right)\)
where $(\alpha,\beta) \in \QQ^2 \setminus \bm{0}$ satisfies \(\alpha^2 + \frac{b}{a}\beta^2=1\) (is a solution to Pell's equation) for which \(\alpha\notin \{\pm 1, \pm \frac{1}{2}, 0\}\) (as in \autoref{bin1sol}).
	What follows can be viewed as an application of the multiplication principle for the generalised Pell's equation~\cite{andreescu2016quadratic}. 
	Observe that if $\bm{v} = (x , y)^\top$ is a solution to $ax^2+by^2=c$, then so is $R\bm{v}$.

We now show how to generate infinitely many rational solutions to $ax^2+by^2=c$ from a single rational solution.
	Assume, towards a contradiction, that $R^{n+k}\bm{v} = R^n \bm{v}$ holds for some $n\geq 0$, $k \geq 1$.
	Therefore, there exists an integer $k$ such that $1$ is an eigenvalue of $R^k$. Equivalently, there exists a root of unity $\omega$ which is an eigenvalue of~$R$.
	We proceed under this assumption.
	
	By construction, the eigenvalues of $R$ are $\omega$ and $\omega^{-1}$. 
		Let $\varphi$ be the argument of $\omega$. 
	  Then the real part of $\omega$,  $\cos(\varphi)$, is equal to $\alpha$ (and thus rational). 
	Since~$\omega$ is a root of unity, $\varphi$ is a rational multiple of~$2\pi$.
	By Niven's theorem~\cite[Corollary 3.12]{niven1956irrat}, the only rational values for $\cos(\varphi)$ are $0$, $\pm \frac{1}{2}$ and $\pm 1$.
	We arrive at a contradiction, as $\alpha$ was carefully picked to avoid these values.
	
In summary, we have shown that $R$ has no eigenvalues that are roots of unity, from which we deduce the desired result.
\end{proof}

\subsection{Solving Isotropic Quadratic Forms} 
We next present an approach to solving \cref{prob:quadr} that uses the theory of representations of quadratic forms. 
First, we prove a lemma concerning the representations of~$0$.

\begin{lemma}\label{rep-not-zero}
	Let $Q(x_1, \dots, x_n) = a_1x_1^2 + \dots + a_nx_n^2$ be an isotropic quadratic form with $a_1, \dots, a_n \neq 0$.
	There exists a representation~$(\alpha_1, \dots, \alpha_n)$ of~$0$; i.e., \(a_1\alpha_1^2 + \dots + a_n\alpha_n^2 = 0\)
	such that $\alpha_1, \dots, \alpha_n \neq 0$.
\end{lemma}
\begin{proof}
	Let $(\beta_1, \dots, \beta_n) \in \QQ^n$ be a representation of~$0$ by~$Q$.
	We further assume that $\beta_1, \dots, \beta_r \neq 0$ while $\beta_{r+1} = \dots = \beta_n = 0$, and $r < n$.
	Moreover, let~$\lambda:= a_r\beta_r^2 + a_{r+1}\beta_{r+1}^2$.
	
	Consider the equation $x^2 + \frac{a_{r+1}}{a_r}y^2 = 1$.
	From~\cref{bin1sol}, it has a rational solution~$(\alpha, \beta)$ such that $\alpha, \beta \neq 0$. 
	This implies $a_r\alpha^2 + a_{r+1}\beta^2 = a_r$.
	The pair $(\beta_r, 0)$ is one solution to $a_rx_r^2 + a_{r+1}x_{r+1}^2 = \lambda$.
	Following the steps in the proof of~\cref{bin-solutions}, we can construct a matrix $R$ for which
	$R \cdot (\beta_r , 0 )^\top = (\alpha \beta_r , \beta \beta_r )^\top$
	where $(\alpha\beta_r, \beta\beta_r)$ is a solution of $a_rx_r^2 + a_{r+1}x_{r+1}^2 = \lambda$ with both components being non-zero.
	Therefore, $(\beta_1, \dots, \beta_{r-1}, \alpha\beta_r, \beta\beta_r, \beta_{r+2}, \dots, \beta_n)$ is an isotropic vector of~$Q$ with fewer zero entries.
	By repeating the process, we obtain an isotropic vector~$(\alpha_1, \dots, \alpha_n)$ as desired.
\end{proof}
	We emphasise that the process of eliminating zeros from the isotropic vector is effective. A similar proof is given in~\cite[p.294, Theorem~8]{borevich1966nt}.

In this discussion, we focus on solving equations of the form~$Q(x_1, \dots, x_d) = c$, where $Q$ is a quadratic form.
As it will be shown later in \cref{sec:qf}, it is always possible to find an equivalent diagonal quadratic form~$D \sim Q$.
Therefore, we restrict our attention to equations of the form $a_1x_1^2+ \dots + a_dx_d^2 = c$.
Assuming~$c \neq 0$, we start by homogenising the equation, and so consider the solutions of
\begin{equation}\label{eq:iso} a_1x_1^2+ \dots + a_dx_d^2 - cx_{d+1}^2 = 0.\end{equation}
In other words, we are searching for a rational isotropic vector of a quadratic form.
\begin{proposition}\label{prop:Isotropic}
	An equation
	\begin{equation}\label{eq:pure-quadr}
		a_1x_1^2+ \dots + a_dx_d^2 = c
	\end{equation}
	has a rational solution different from $(0, \dots, 0)$
	if and only if 
	the quadratic form~$Q = a_1x_1^2 + \dots + a_dx_d^2 - cx_{d+1}^2$
	has an isotropic vector.
\end{proposition}
\begin{proof}
	For $c=0$, the statement is a recitation of a definition.
	We continue under the assumption~$c \neq 0$.
	Recall from~\cref{rep-not-zero} that if the form~$Q$ is isotropic, then there is an isotropic vector~$(\alpha_1, \dots, \alpha_{d+1})$ with $\alpha_i \neq 0$ for all $i \in \{1, \dots, d+1\}$.
	Therefore, we can find a non-zero solution $(\alpha_1/\alpha_{d+1}, \dots, \alpha_d/\alpha_{d+1})$ to~\cref{eq:pure-quadr}.
	Conversely, if \eqref{eq:pure-quadr} has a non-trivial solution $(\beta_1, \dots, \beta_d)$, it follows that $(\beta_1, \dots, \beta_d, 1)$ is an isotropic vector for~$Q$.
\end{proof}

\subsection{Finding Isotropic Vectors}
\cref{prop:Isotropic} implies that
solving \cref{prob:dioph}, and hence also loop synthesis in
\cref{prob:synthesis}, 
requires detecting whether a certain quadratic form is
isotropic. 
Effective isotropy tests are known for quadratic forms~$Q(x_1, \dots, x_{d+1})$ as in~\cref{eq:iso}.
A more difficult task is the problem of finding an isotropic vector for such a form. 

The abstract arithmetic techniques employed in finding an isotropic vector are beyond the scope of this paper; however, we give a brief overview of the computational task and a number of references to the literature in \cref{app:appendix-isotropy}.
Our takeaways from the theory
are the following functions:
	\begin{itemize}
		\item \textproc{isIsotropic}: a function that, given an indefinite quadratic form over the rationals as an input, determines whether the input is isotropic and duly returns the answers \textsc{yes} and \textsc{no} (as appropriate).
		\item \textproc{findIsotropic}: a function that accepts isotropic quadratic forms over the rationals as inputs and returns an isotropic vector for each such form.
		\item \textproc{solve}: 
	a function that takes \cref{eq:pure-quadr} as an input and returns a non-zero solution if the form~$a_1x_1^2+\dots+a_dx_d^2-cx_{d+1}^2$ is isotropic; otherwise \textproc{solve} returns ``\textsc{no solutions}''.
		The function \textproc{solve} calls both \textproc{isIsotropic} and \textproc{findIsotropic}, as shown in \autoref{alg:isotropic}.
	\end{itemize}
	
We note the \textproc{solve} subroutine in the sequel: the function \textproc{linLoop} defined in \cref{alg:lin}, calls on \textproc{solve}; and, in turn, the function \textproc{linLoop} is called %
by the procedure in~\cref{sec:qe}.

\section{Quadratic Forms: Linear Loops}\label{sec:qf}
The core of this section addresses equations, and hence loop invariants, that involve quadratic forms.
The equations (invariants) of this section do not have a linear part; they are quadratic forms equated to constants; that is, equations of the form
\begin{equation}\label{qf-c}
	Q(x_1, \dots, x_d) = c,
\end{equation}
where $Q$ is an arbitrary $d$-ary quadratic form with rational coefficients, $c$ is a rational number.

The main result of this section is the following theorem, which establishes a decision procedure that can determine if a given quadratic invariant admits a linear loop and, if so, constructs that loop.
\begin{theorem}[Linear Loops for Quadratic Forms]\label{loop-for-qf}
	There exists a procedure that, given an equation $Q(x_1, \dots, x_d) = c$ of the form~\eqref{qf-c},
	decides whether a non-trivial linear loop satisfying $Q(x_1, \dots, x_d) = c$ exists and, if so,
	synthesises a loop.
\end{theorem}

We prove~\cref{loop-for-qf} in several steps. The first of them is to diagonalise the quadratic form~$Q$ and thus reduce to~\cref{qf-c} without mixed terms on the left-hand side.

\subsection{Rational Diagonalisation} %
\label{sec:rat-diag}
A rational quadratic form can be diagonalised %
by an invertible change of variables with only rational coefficients. 
\begin{proposition}\label{rat-diag}
	Let~$Q$ be a (possibly degenerate) $d$-ary quadratic form. There exists an equivalent quadratic form~$D$ with a diagonal matrix~$A_D \in \QQ^{d\times d}$, i.e., $Q \sim D$. Furthermore, \(
		A_D = \sigma^\top A_Q\sigma\) holds with $\sigma \in \GL_d(\QQ)$.
\end{proposition}
A diagonalisation algorithm is described in~\cite[Algorithm 12.1]{schaum2009}, 
see also ``diagonalisation using row/column operations'' in~\cite[Chapter 7, 2.2]{treil2004laWrong}. %
The idea, as presented in~\cite{treil2004laWrong}, is to perform row operations on the matrix~$Q$. Different from the usual Gauss--Jordan elimination, the analogous column operations
are performed after each row operation. We emphasise that the change-of-basis matrix~$\sigma$ is invertible as a product of elementary matrices.	
\begin{remark}[Degeneracy]
	Let $A_D := \diag{a_1, \dots, a_d}$ be the diagonal matrix of the quadratic form~$D$ as in~\cref{rat-diag}. The product $a_1\cdots a_d$ is zero if and only if the initial quadratic form~$Q$ is degenerate.
\end{remark}

\begin{proposition}\label{reduce-to-diag}	
	Let $Q_1$ and $Q_2$ be two equivalent $d$-ary quadratic forms.
	If there exists a linear loop~$\loopie = \langle M, \bm{s}\rangle$ with invariant $Q_2 = c$ for a constant~$c \in \QQ$,
	then $Q_1 = c$ is an invariant of the linear loop~$\loopie' = \langle \sigma M \sigma^{-1}, \sigma \bm{s} \rangle$. Here, $\sigma \in \GL_d(\QQ)$ is a change-of-basis matrix 
	such that $Q_2 (\bm{x}) = Q_1 (\sigma\cdot\bm{x})$.
\end{proposition}
\begin{proof}
	If $(M^n\bm{s})^\top A_{Q_2}(M^n\bm{s}) = c$ for all $n\geq 0$, then
	\begin{multline*}
		\left((\sigma M \sigma^{-1})^n\sigma \bm{s}\right)^\top A_{Q_1} \left((\sigma M \sigma^{-1})^n\sigma \bm{s}\right) = \left(\sigma M^n\bm{s}\right)^\top A_{Q_1}\left(M^n\bm{s}\right) \\
		= \bm{s}^\top (M^n)^\top \sigma^\top A_{Q_1}\sigma M^n\bm{s} = \bm{s}^\top (M^n)^\top A_{Q_2} M^n\bm{s} = (M^n\bm{s})^\top  A_{Q_2} M^n\bm{s} = c
	\end{multline*} 
for all $n\geq 0$ as well.
We emphasise that $\sigma$ is a bijection from $\QQ^d$ to itself, so
the reduction described here preserves the infiniteness of loop orbits.
\end{proof}
We conclude from~\cref{rat-diag,reduce-to-diag} that for a general quadratic form~$Q$, 
a linear loop with an invariant $Q(\bm{x})=c$ exists if and only if 
a linear loop exists for an invariant~$D(\bm{x})=c$, where~$D$ is an equivalent diagonal form. 

\subsection{Diagonal Quadratic Forms} 
In this subsection we consider diagonal quadratic forms $a_1x_1^2 + \dots + a_d x_d^2 = c$, where $a_1, \dots, a_d, c \in \QQ$ as in \cref{eq:pure-quadr}.
If the equation is homogeneous; that is, $c = 0$, then loop synthesis reduces to the problem of searching for a rational solution~$\bm{\alpha} = (\alpha_1, \dots, \alpha_d)$.
Indeed, a loop with a matrix $\lambda \cdot I_d$ (scaling each variable by~$\lambda \in \QQ \setminus\{-1,0,1\}$) and the initial vector $\bm{\alpha}$ is a non-trivial linear loop satisfying the invariant $Q(\bm{x}) = 0$.

From \cref{sec:eqsolv}, we know how to generate a solution (or
prove there is no solution) to~\cref{eq:pure-quadr} in its general form,
also with~$c \neq 0$. The bottleneck of
loop synthesis in \cref{prob:synthesis} is thus finding an update matrix~$M$ for the linear loop.
En route to solving this issue, we state the following corollary of \autoref{bin-solutions}.

\begin{corollary}\label{bin-loop-or-nothing}
	If an equation $ax^2+by^2=c$ with $a, b \in \QQ \setminus{0}$ has infinitely many rational solutions different from $(0,0)$, then there exists a non-trivial linear loop with polynomial invariant $ax^2+by^2=c$.
\end{corollary}

\begin{proof}
We use the construction in the proof of \autoref{bin-solutions}, which demonstrates that the orbit of the linear loop $\langle R, \bm{v}\rangle$ is infinite with polynomial invariant $ax^2+by^2=c$.
\end{proof}

\begin{proof}[Proof of \cref{loop-for-qf}]
	Due to \cref{reduce-to-diag}, we can consider an equation of the form~\eqref{eq:pure-quadr}: \[a_1x_1^2 + \dots + a_d x_d^2 = c.\]
	We describe the loop synthesis procedure in this case. If $d=1$, the equation only has finitely many solutions, hence any loop for~\cref{eq:pure-quadr} is trivial. Hereafter we  assume that $d\geq 2$.
	
	In order to generate an initial vector of the loop for~\cref{eq:pure-quadr}, we exploit the results of~\cref{sec:eqsolv}.
	Either \cref{eq:pure-quadr} has no rational solutions and hence no loop exists, or
	we effectively construct a solution $\bm{\alpha} = (\alpha_1, \dots, \alpha_d) \in \QQ^d$
	using procedure \textproc{solve}.
	Recall that we can guarantee $\alpha_i \neq 0$ for all $i \in \{1,\dots, d\}$ due to~\cref{rep-not-zero}.
	
	Note that some of the coefficients $a_i$, $i \in \{1, \dots, d\}$, may be zero if the original quadratic form $Q$ is degenerate.
	We have to consider the case when all coefficients but one are~0, separately.
	That is, $a_1x_1^2 + 0x_2^2 + \dots + 0x_d^2 = c$.
	For this form, a solution exists if and only if $c/a_1$ is a square of a rational number.
	Subsequently, if a solution $\bm{\alpha}$ is found, set $M:= \diag{1, 2, \dots, 2}$ to be a diagonal update matrix.
	Since $d \geq 2$, we guarantee that the orbit of the linear loop~$\langle M, \bm{\alpha}\rangle$ is infinite.
	
	Without loss of generality, we now assume $a_1 \neq 0$ and $a_2 \neq 0$.
	Define $\gamma := a_1\alpha_1^2+a_2\alpha_2^2$, then the equation
\(		a_1x_1^2+a_2x_2^2 = \gamma\)
has a non-trivial solution $(\alpha_1, \alpha_2)$.
	
	From \cref{bin-loop-or-nothing}, there exists a matrix \(R\in\QQ^{2\times 2}\) that preserves the value of the quadratic form \(a_1x_1^2+a_2x_2^2\). 
	This matrix can be constructed as in the proof of~\cref{bin-loop-or-nothing} by considering the equation $x_1^2+\frac{a_2}{a_1}x_2^2 = 1$. 
Let \(M\) be the matrix given by the direct sum
	\[ R \oplus I_{d-2} =
		\left(\begin{array}{c|c} 
			R & 0 \\
			\hline
			0 & I_{d-2}
		\end{array}\right)
	\]
	where $I_n$ is an identity matrix of size $n$.
	
	A desired loop is $(M, \bm{\alpha})$ as for each $n \geq 0$, $M^n\bm{\alpha}$ satisfies \cref{eq:pure-quadr}.
	The loop is non-trivial because its orbit, restricted to $x_1, x_2$, is infinite.
\end{proof}

The process of synthesising a loop for the quadratic invariant $Q(x_1, \dots,
x_n) = c$ is summarised in \cref{alg:lin}.
The algorithm starts with a diagonalisation step, proceeds with finding a loop for an equation of the form~\eqref{eq:pure-quadr}, and applies the inverse transformation to obtain a linear loop for the initial invariant.
Whenever \cref{alg:lin} returns a loop, this loop is linear.

\begin{algorithm}[t] 
		\caption{Synthesise a linear loop satisfying a given quadratic form equation}		\label{alg:lin}
	\begin{algorithmic}[1]
	\Require %
	quadratic form \(Q\) in \(d\) variables and \(c\in\QQ\).  Assert \(d\geq 2\).
	\Function{linLoop}{$Q,s$}
	\State $\langle M, s \rangle :=$ \textsc{undefined}.
		\State compute a rational diagonalisation for~$Q(\bm{x})$: 
		a $\sigma \in \GL_d(\QQ)$ such that $Q'(\sigma\bm{x}) = Q(\bm{x})$ with~$Q'$ a diagonal quadratic form
		\State rewrite the equation $Q' = c$ as $a_1x_1^2 + \dots + a_dx_d^2 =c$ with $a_1, \dots, a_r \neq 0$ and $a_{r+1} = \dots = a_d = 0$
 		\State let $\bm{\alpha}:= (\alpha_1, \dots, \alpha_r, 1, \dots, 1)^\top \in \QQ^d$, where $(\alpha_1, \dots, \alpha_r):=$  \textproc{solve}$(a_1,\dots,a_r,c)$
 		
 		\Comment{for \textproc{solve} see \autoref{alg:isotropic}.}
		\If{$r=1$ and $\bm{\alpha} \neq $ ``\textsc{no solutions}'' }
			\State $M:=\diag{1,2, \dots, 2}$.
		\ElsIf{ $\bm{\alpha}=$ ``\textsc{no solutions}''  }
				\State \Return ``\textsc{no loop}''. 
		\Else
			\State compute a solution $(y_1, y_2)$ of $x_1^2+\frac{a_2}{a_1}x_2^2 = 1$.
			\Comment{see~\cref{bin1sol}}
			\State \(M:= R \oplus I_{d-2}\), where $R = \left(\begin{smallmatrix}
			y_1 & -\frac{a_2}{a_1} y_2\\
			y_2 & y_1\\
		\end{smallmatrix}\right)$.
		\EndIf
		\State \Return $\langle \sigma^{-1}M\sigma, \sigma^{-1}\bm{\alpha} \rangle$.
	\EndFunction
	\end{algorithmic}
\end{algorithm}

\section{Arbitrary Quadratic Equations: Affine Loops}
\label{sec:qe}
In this section, we leave the realm of quadratic forms and consider general quadratic invariants that may have a linear part. 
Any quadratic equation can be written in terms of a quadratic form~$Q$, a linear form~$L$, and a constant term~$c$:%
\begin{equation}\label{eq:quadr}Q(x_1, \dots, x_d) + L(x_1, \dots, x_d) = c.\end{equation}

On our way to a complete solution %
of~\cref{prob:synthesis} for arbitrary quadratic equations, we carefully analyse~\cref{eq:quadr}.
A standard technique (see e.g.~\cite[Proposition 1]{grunewald1981solveZ}) allows to reduce~\cref{eq:quadr} with a non-degenerate quadratic form $Q$ to~\cref{qf-c} considered in~\cref{sec:qf}.
We now give the details of this reduction and describe how to synthesise an affine loop for an invariant~\eqref{eq:quadr} in the non-degenerate case.
Subsequently, we close the gap by discussing the case when $Q$ is
degenerate. Using \cref{rem:LinAffLoop}, our results on affine loop synthesis
imply then linear loop synthesis.

\subsection{Non-Degenerate Quadratic Forms} \label{ssec:nondegeneratequadratic}

For convenience, we rewrite the equation in the matrix-vector form:
\(\bm{x}^\top  A_Q \bm{x} + \bm{b}^\top  \bm{x} - c = 0\).
Here, $A_Q$ is the non-singular matrix of the quadratic form $Q$, and $\bm{b}$ is the vector of the linear form.
Let $\delta := \det{A_Q} \neq 0$ and $C$ be the cofactor matrix of $A_Q$, i.e.,
\(A_Q \cdot C = C \cdot A_Q = \delta \cdot I_d\).
We further define $\bm{h}:= C \cdot \bm{b}$ and $\tilde{c} = 4\delta^2c + Q(\bm{h})$. 
It can be checked directly that \begin{equation}\label{legendre-transf}
	Q(2\delta \cdot \bm{x} + \bm{h}) = \tilde{c} \Leftrightarrow Q(\bm{x}) + L(\bm{x})= c.
\end{equation}
In words, every equation of the form~\cref{eq:quadr} can be reduced to an equation of the form $Q(\bm{y}) = \tilde{c}$ \emph{by an affine transformation} $f$ that maps each $\bm{x} \in \QQ^d$ to $2\delta \cdot \bm{x} + \bm{h} \in \QQ^d$.
As such, this means that solutions of \cref{eq:quadr} under the non-degeneracy assumption are in a one-to-one correspondence with representations of~$\tilde{c}$ for~$Q$.
\begin{proposition}\label{loop-for-qe-nondeg}
	Let $Q$ be a non-degenerate quadratic form and $L$ a linear form, both in $d \geq 2$ variables. 
	Define $\delta:= \det(A_Q)$, $\bm{h}$ and $\tilde{c}$, as in the discussion above.
	The following are equivalent:
	\begin{enumerate}
		\item There exists a linear loop $\langle M, \bm{s} \rangle$ satisfying the invariant $Q(\bm{x}) = \tilde{c}$.
		\item There exists an affine loop
		\[\langle M, \frac{1}{2\delta}\left(\bm{s}- \bm{h}\right),\frac{1}{2\delta}\left(M - I_d\right)\bm{h}\rangle,\] 		
		satisfying the invariant $Q(\bm{x}) + L(\bm{x}) = c$.
	\end{enumerate}
\end{proposition}
\begin{proof}
	Start with the first assumption.
	For all $n \geq 0$, it holds $Q(M^n\bm{s}) = \tilde{c}$.
	Equivalently, \[Q(f^{-1}\left( M^n\bm{s}\right)) + L(f^{-1}\left( M^n\bm{s}\right))= c, \; \text{or} \;
	Q\left(\frac{1}{2\delta}(M^n\bm{s} - \bm{h})\right) + L\left(\frac{1}{2\delta}(M^n\bm{s} - \bm{h})\right) = c\]
	for all $n \geq 0$.	

	On the other hand, let $\bm{x}(n)$ be the variable vector after the $n$th iteration of an affine loop from the statement.
	We prove by induction that $\bm{x}(n) = \frac{1}{2\delta}(M^n\bm{s} - \bm{h})$.
	The base case is true since the initial vector of the affine loop is $\frac{1}{2\delta}(\bm{s} - \bm{h}) = \frac{1}{2\delta}(M^0\bm{s} - \bm{h})$.
	Now, assume that $\bm{x}(k) = \frac{1}{2\delta}(M^k\bm{s} - \bm{h})$ for an arbitrary $k \geq 0$.
	Then, by applying the loop update once, we have
	\begin{multline*}
		\bm{x}(k+1) =M \cdot \left(\frac{1}{2\delta}(M^k\bm{s} - \bm{h})\right) + \frac{1}{2\delta}\left(M - I_d\right)\bm{h} \\
		= \frac{1}{2\delta}\left(M^{k+1}\bm{s} - M\bm{h} + M\bm{h} - \bm{h} \right) =\frac{1}{2\delta}\left(M^{k+1}\bm{s} - \bm{h} \right),
	\end{multline*}
and the inductive step has been shown.	
By the above work, we conclude that $Q(\bm{x}(n)) + L(\bm{x}(n)) = c$ holds for all~$n\geq 0$.
\end{proof}

\begin{example}\label{ex:nondeg-affine}
	Consider an invariant $p(x,y) := x^2+y^2-3x-y = 0$.
		After an affine change of coordinates $f(x,y) = (2x-3,2y-1)$, it becomes
		$x^2+y^2=10$ 
		(that corresponds to $\delta = 1$, $\bm{h} = (-3,-1)^\top$, $\tilde{c} = 10$).
		There exists a linear loop for this equation:
		\[M = \begin{pmatrix}
			\frac{3}{5} & -\frac{4}{5}\\
			\frac{4}{5} & \frac{3}{5}
		\end{pmatrix} \text{ and } \bm{s} = 
	\begin{pmatrix}	1 \\ -3 \end{pmatrix}.\]
	Next, compute the components of an affine loop.
	The update matrix is~$M$, whilst the
	initial and translation vectors are
	\[\frac{1}{2} \left[\begin{pmatrix}	1 \\ -3 \end{pmatrix} - \begin{pmatrix}	-3 \\ -1 \end{pmatrix}\right] = \begin{pmatrix}	2 \\ -1 \end{pmatrix} \hfill \text{and} \hfill \frac{1}{2} \left[\begin{pmatrix}
		\frac{3}{5} & -\frac{4}{5}\\
		\frac{4}{5} & \frac{3}{5}
	\end{pmatrix} -  \begin{pmatrix}
	1& 0\\
	0 & 1
\end{pmatrix}\right]\begin{pmatrix}	-3 \\ -1 \end{pmatrix}= \begin{pmatrix}	1 \\ -1 \end{pmatrix},
\]	
respectively.
The resulting affine loop is non-trivial with invariant $p(x,y) = 0$ due to~\cref{loop-for-qe-nondeg}:
		\begin{equation*}
	 \begin{pmatrix}
		x \\ y
	\end{pmatrix}\leftarrow  \begin{pmatrix}
	2\\ -1
	\end{pmatrix};\ \textrm{while}\ \star\ \textrm{do}\ \begin{pmatrix}
			x \\ y
		\end{pmatrix}\leftarrow 
	\begin{pmatrix}
		\frac{3}{5}x-\frac{4}{5}y + 1 \\ \frac{4}{5}x+\frac{3}{5}y - 1
	\end{pmatrix}.
\end{equation*}
\end{example}
\subsection{Degenerate Quadratic Forms} \label{ssec:degeneratequad}
Let $r < d$ be the rank of~$A_Q$.
There exist $k:=d-r$ linearly independent vectors~$\bm{v}_1, \dots, \bm{v}_k\in \QQ^d$ such that
$A_Q \cdot \bm{v}_i = \bm{0}$.
Construct a matrix $\tau \in \GL_d(\QQ)$ such that $\bm{v}_1, \dots, \bm{v}_k$ constitute its first columns.
It follows that every non-zero entry $(M)_{ij}$ of a matrix $M:= \tau^\top A_Q\tau$ %
is located in the bottom right corner, that is, $i >k$ and $j > k$.
We rewrite $Q(\tau\bm{x}) = \tilde{Q}(x_{k+1}, \dots, x_d)$ and $L(\tau\bm{x}) =  \tilde{L}(x_{k+1}, \dots, x_d) + \lambda_1x_1 + \dots + \lambda_kx_k$ in~\cref{eq:quadr}.
Now we have:
\begin{equation} \label{eq:QLA} \tilde{Q}(x_{k+1}, \dots, x_d) +  \tilde{L}(x_{k+1}, \dots, x_d) = c - \lambda_1x_1 - \dots - \lambda_kx_k,\end{equation}
where $\tilde{Q}$ is a non-degenerate quadratic form of~$r$ variables.

In the rest of this subsection, we are concerned with finding an affine loop satisfying~\cref{eq:QLA}.
We emphasise that such a loop~$\langle M, \bm{s}, \bm{t} \rangle$ exists if and only if~$\langle \tau M \tau^{-1}, \tau\bm{s}, \tau\bm{t} \rangle$ satisfies~\cref{eq:quadr}. The proof is due to~$\tau$ inducing an automorphism of~$\QQ^d$, cf.~\cref{reduce-to-diag}.

If $\lambda_1 = \dots = \lambda_k = 0$, we have arrived at an instance of~\cref{eq:quadr} with a non-degenerate quadratic form and fewer variables.
Let $\delta$ be the determinant of~$\tilde{Q}(x_{k+1}, \dots, x_d)$ and, as in the non-degenerate setting, define an affine transformation~$f$ on the subset of variables~$\{x_{k+1}, \dots, x_d\}$.
The constant~$\tilde{c}$ and the vector~$\bm{h} \in \QQ^r$ are defined similarly to their non-degenerate setting counterparts.

After the change of coordinates that corresponds to~$f$, we have
\begin{equation}\label{eq:with0}
	0x_1^2+\dots+0x_k^2+\tilde{Q}(x_{k+1}, \dots, x_d) = \tilde{c}.
\end{equation}
Recall (e.g.~from the proof of~\cref{loop-for-qf}) that once \cref{eq:with0} with $k\geq 1$ has a solution, there is a non-trivial linear loop satisfying the polynomial invariant defined by the equation.
Now, let $\langle M,\bm{s}\rangle$ be a linear loop for~\cref{eq:with0},
where $\bm{s} = (s_1, \dots, s_d)^\top$.
In fact, one can assume $M:=\diag{2,\dots, 2,1, \dots, 1}$ with $k$ twos and $r$ ones.
Define $\bm{s'}:=\frac{1}{2\delta}\left(\bm{s} - \left(\begin{smallmatrix}
	\bm{0} \\ \bm{h}
\end{smallmatrix}\right)\right)$.
It is not hard to see that a non-trivial \emph{linear} loop~$\langle M, \bm{s}' \rangle$ satisfies $\tilde{Q}(x_{k+1}, \dots, x_d) +  \tilde{L}(x_{k+1}, \dots, x_d) = c$
if and only if $\tilde{Q}(x_{k+1}, \dots, x_d) = \tilde{c}$ has a solution~$(s_{k+1}, \dots, s_d)$.

From now on, we assume that $k\geq 1$ is the number of non-zero $\lambda_i$'s on the right-hand side of~\cref{eq:QLA}. We show next that the loop synthesis question has a positive answer.
\begin{proposition}[Affine Loops for Quadratic Forms]\label{loop-for-qe-deg}
	 Given a quadratic equation of the form~\eqref{eq:QLA}, there exists a non-trivial affine loop in variables $x_1, \dots, x_d$ for which said equation is a polynomial invariant.
\end{proposition}
\begin{proof}
	Since $k \geq 1$ and $\lambda_1 \neq 0$, the right-hand side $c - \sum_{i=1}^k \lambda_i x_i$ represents every rational number.
	Set the values of $x_{k+1}, \dots, x_d$ to some fixed values $\bm{\alpha} = (\alpha_1, \dots, \alpha_{d-k})$ such that $\bm{\alpha} \neq \bm{0}$
	and solve the equation for $x_1, \dots, x_k$ attaining a vector of values $\bm{\beta} = (\beta_1, \dots, \beta_k)$. 
	We have $ \tilde{Q}(\bm{\alpha}) +  \tilde{L}(\bm{\alpha})= A(\bm{\beta})$, where $A(x_1, \dots, x_k) := c - \lambda_1 x_1 - \dots - \lambda_k x_k$.

	We introduce the following case distinction.
	\begin{description} 
		\item[\hypertarget{item:1}{Case 1.}] $k > 1$; 
		\item[\hypertarget{item:2}{Case 2.}] $r > 1$ and so $\tilde{Q}$ is a non-degenerate quadratic form of at least~$2$ variables; 
		\item[\hypertarget{item:3}{Case 3.}] $r = 1$ and $k = 1$; that is, \cref{eq:QLA} has the form $ax^2+bx = c - dy$, $d\neq 0$.
	\end{description}
	In the rest of the proof, we show that for all these cases, a non-trivial affine loop satisfies~\cref{eq:QLA} and hence, the invariant of~\cref{eq:quadr}. Moreover, in Cases \hyperlink{item:1}{1} and \hyperlink{item:2}{2} there exist \emph{linear} loops of this sort.

	In \hyperlink{item:1}{Case 1}, we focus on the vector~$\bm{\beta}$ computed in the previous step.
	Without loss of generality, $(\beta_1, \beta_2) \neq (0,0)$.
	We construct a linear loop that preserves the values of all variables but $\beta_1, \beta_2$.
	To this end, it suffices to notice that a linear transformation of~$\QQ^2$ defined by $(x_1, x_2) \mapsto (2x_1, -\frac{\lambda_1}{\lambda_2}x_1+x_2)$ preserves the value of $\lambda_1x_1 + \lambda_2x_2$.
	The desired linear loop has initial vector~$\bm{s} = (\bm{\beta}, \bm{\alpha} )^\top $ and an update matrix~$M = \left(\begin{smallmatrix}2 & 0\\-\frac{\lambda_1}{\lambda_2} & 1 \end{smallmatrix}\right) \oplus I_{d-2}$.
	
	Let us turn to \hyperlink{item:2}{Case 2} and focus on vector~$\bm{\alpha}$. 
	Clearly, we can now assume $k = 1$.
	Without loss of generality, we shall assume that $\beta_1 \neq 0$.
	Consider the equation
	\begin{equation*}\label{nondeg-eq-trans}
		\tilde{Q}(\bm{x}) + \tilde{L}(\bm{x}) = A(\beta_1)
	\end{equation*}
	over the variables~$\bm{x}$, with~$A(y) = c - \lambda_1y$.
	 Using~\cref{legendre-transf}, we argue that its solutions are related to the representations of a certain number $\tilde{c}$ by~$ \tilde{Q}$.
	  We compute~$\delta$, $\tilde{c}$, $\bm{h}$ for the non-degenerate quadratic form~$\tilde{Q}$, linear form~$\tilde{L}$ and constant~$A(\beta_1)$ such that
	$\tilde{Q}(2\delta \cdot \bm{x} + \bm{h}) = \tilde{c}$.
	From~\cref{loop-for-qf} and, more specifically, its proof, observe that there exists a non-trivial \emph{linear} loop satisfying $ \tilde{Q}(\wc) = \tilde{c}$.
	Indeed, there exists at least one solution of this equation, namely~$f(\bm{\alpha})$.
	Let $\langle M, \bm{s}\rangle$ be a linear loop satisfying $\tilde{Q}(\wc) = \tilde{c}$ with matrix $M\in\QQ^{r\times r}$.
	\cref{loop-for-qe-nondeg} shows that an affine loop 
	\[\mathcal{A}:= \langle M, \frac{1}{2\delta}\left(\bm{s}- \bm{h}\right), \frac{1}{2\delta}\left(M - I_r\right)\bm{h}\rangle,\] 		
	satisfies the invariant $\tilde{Q}(\bm{x}) + \tilde{L}(\bm{x}) = A(\beta_1)$.
	The sequence~$\langle\bm{x}(n)\rangle_{n=0}^\infty$ of $\mathcal{A}$'s variable vectors can be expressed in terms of an augmented matrix  (see \(M'\) in~\cref{sec:prelims})
	associated with the affine transformation \(\bm{x}\mapsto M \bm{x} + \bm{t}\), where $\bm{t} = \frac{1}{2\delta}\left(M - I_r\right)\bm{h}$ and $\bm{s}' = \frac{1}{2\delta}\left(\bm{s}- \bm{h}\right)$:
	\[
	\begin{pmatrix} 1 \\ \bm{x}(n) \end{pmatrix} =
	\left(\begin{array}{c|c} 1 & 0_{1,r} \\ \hline \bm{t}  & M	\end{array}\right)^{\!n}\!\! \begin{pmatrix} 1 \\ \bm{s}' \end{pmatrix}, 
	\]
	satisfies $ \tilde{Q}(\bm{x}) + \tilde{L}(\bm{x}) = A(\beta_1)$ for all~$n\geq 0$. Then, %
	\[
	\begin{pmatrix}y(n) \\ \bm{x}(n) \end{pmatrix} =
	\left(\begin{array}{c|c} 1 & 0_{1,r} \\ \hline \frac{1}{\beta_1}\bm{t}  & M	\end{array}\right)^{\!n} \!\! \begin{pmatrix} \beta_1 \\  \bm{s}' \end{pmatrix} 
	\]
	satisfies $\tilde{Q}(\bm{x}(n)) + \tilde{L}(\bm{x}(n)) = A(y(n))$ as in~\cref{eq:QLA} for all~$n\geq 0$. 
	We denote by \(M_\beta\) the ${d}$-dimensional square matrix in the preceding displayed equation.
Observe that $\langle M_\beta, (\beta_1, \bm{s}')^\top \rangle$ is a \emph{linear} loop satisfying the invariant of~\cref{eq:quadr}.
	
	Finally, we come to the special case, \hyperlink{item:3}{Case 3}, that considers quadratic equations of the form \(ax^2+bx = c - dy\) where $d \neq 0$.
	It suffices to observe that an affine transformation of~$\QQ^2$ defined by $(x,y) \mapsto (2x, 2\frac{b}{d} x +4 y -3 \frac{c}{d})$ preserves the equation \(ax^2+bx = c - dy\).
	We conclude that  \(ax^2+bx = c - dy\) is a polynomial invariant of the affine loop with initial vector~$(1,  \frac{c-a-b}{d})^\top $, translation vector~$( 0, -3\frac{c}{d})^\top $, and update matrix \(\left(\begin{smallmatrix}
		2 & 0 \\
		2\frac{b}{d} & 4
	\end{smallmatrix}\right)\).
\end{proof}
\subsection{The Procedure: Affine Loop Synthesis for Quadratic Invariants}
\begin{theorem}[Affine Loops for Quadratic Equations]\label{loop-for-qe}
	There exists an effective procedure that, given a quadratic
        equation (i.e. invariant) 
	\[Q(x_1, \dots, x_d) + L(x_1, \dots, x_d) = c,\]
	decides whether a non-trivial affine loop satisfying it exists and, if so,
	synthesises a loop.
\end{theorem}
The theorem is essentially proved in \cref{loop-for-qe-nondeg,loop-for-qe-deg}.
If the quadratic form is non-degenerate, \cref{loop-for-qe-nondeg} reduces the search for an affine loop to the search for a linear loop satisfying $Q(x_1, \dots, x_d) = \tilde{c}$.
The solution of this problem was given in~\cref{loop-for-qf}.
If the quadratic form is degenerate, we consider \cref{eq:QLA}. 
If at least one of the $\lambda_i$'s is non-zero, a loop exists, as shown by the ad~hoc constructions of~\cref{loop-for-qe-deg}.
In two of the three cases there, the loop is not just affine, but linear.
Otherwise, if all of the $\lambda_i$'s are zero, we obtain a linear loop by essentially testing whether a solution to an equation $\tilde{Q}(x_1, \dots, x_d) = \tilde{c}$ exists. 
Finally, in order to obtain an affine loop satisfying the original equation, we apply transformation~$\tau$ to the loop synthesised for~\cref{eq:QLA}.

The synthesis procedure is summarised in \cref{alg:aff} (see \cref{app:appendix-mainalgorithm}). %
By analysing the algorithm, one can argue that a negative output implies that~\cref{eq:quadr} has no solutions.
The problem of \emph{deciding} whether a loop exists for a given
invariant, as in \cref{prob:synthesis} and opposed to the synthesis of numerical values, is thus solved as follows.
\begin{corollary} \phantomsection \label{cor:loop-exists}
	Let $Q$ be a quadratic form, $L$ a linear form over variables~$\bm{x} = (x_1, \dots, x_d)$.
	\begin{enumerate}
		\item A non-trivial affine loop satisfying the quadratic equation $Q(\bm{x}) + L(\bm{x}) = c$ exists if and only if the equation has a rational solution different from $\bm{x} = \bm{0}$.
		\item A non-trivial linear loop satisfying the equation $Q(\bm{x}) = c$ exists if and only if the equation has a rational solution different from $\bm{x} = \bm{0}$.
	\end{enumerate}
\end{corollary}

\begin{example}
	Let $-11x^2+y^2-3z^2+2xy-12xz + x +z = -1$ be a quadratic invariant in~$3$ variables.
	The quadratic form~$Q(x,y,z) = -11x^2+y^2-3z^2+2xy-12xz$ is degenerate with rank $r = 2$ and so 
	we can compute $\tau = \left( \begin{smallmatrix}
		-1 & 0 & 0\\
		1 & 3 & 0\\
		2 & 0 & 3
	\end{smallmatrix}\right)$
such that $\tau^\top  A_Q \tau = \diag{0,9,-27}$ is the matrix of an equivalent form.
We have $Q(\tau \bm{x}) = \tilde{Q}(y, z) = 9y^2-27z^2$.
For the linear part, $L(x, y, z) = x+z$, the change of coordinates results in
$L(\tau \bm{x}) = \tilde{L}(y,z) + x = 3z + x$.
Continue with the equation of the form~\eqref{eq:QLA}:
\(9y^2-27z^2+3z = - 1 - x\).
Here, $\lambda_1 = 1$, and so we set $(y, z)$ to $(\alpha_1, \alpha_2) = (\frac{1}{3}, 0)$ and find a solution for $x$: $\beta_1 = -2$.
Next, find an affine transformation~$f$ associated with $9y^2-27z^2+3z=1$. We have $\delta =243$, $\bm{h} = (0, 27)^\top $ and $\tilde{c} = 216513$.
The solutions of $ 9y^2-27z^2+3z = 1$ are exactly the solutions of $9y^2-27z^2 = 216513$ under the action of~$f$.

Using the \textproc{linLoop} procedure, we find a linear loop $\langle M, \bm{s} \rangle$ for the invariant $9y^2-27z^2 = 216513$ with
\(M = \left(\begin{smallmatrix} 2 & 3\\
	1 & 2 \end{smallmatrix}\right)\) and \(\bm{s} = (-162, 27)^\top\). 
Therefore, an affine loop
\[\mathcal{A}:= \langle M, \frac{1}{2\delta}\left(\bm{s}- \bm{h}\right), \frac{1}{2\delta}\left(M - I_2\right)\bm{h}\rangle;\] 	
that is, an affine loop with augmented matrix~$M'$ and initial vector~$\bm{s}'$ given by
\[M' = \left(\begin{array}{c|c} 1 & 0_{1,r} \\ \hline \frac{1}{2\delta}\left(M - I_2\right)\bm{h}  & M	\end{array}\right) = \begin{pmatrix}
	1 & 0 & 0\\
	-1/6 & 2 & 3\\
	-1/18 & 1 & 2
\end{pmatrix}
\quad \text{and} \quad
\bm{s}' = \frac{1}{2\delta}\left(\bm{s}- \bm{h}\right) = \begin{pmatrix}
	\frac{1}{3} \\ 0
\end{pmatrix}\!\!,
\]
satisfies the invariant $9y^2-27z^2+3z = 1$.
Consequently, a \emph{linear} loop with update matrix
\[M_\beta:= \begin{pmatrix}
	1 & 0 & 0\\
	1/12 & 2 & 3\\
	1/36 & 1 & 2
\end{pmatrix}\] and initial vector $( -2, 1/3, 0)^\top $ satisfies the invariant \(9y^2-27z^2+3z = - 1 - x\).
We conclude by applying transformation~$\tau$:
a linear loop with matrix
\[\tau  M_\beta  \tau^{-1} = \begin{pmatrix}
	1 & 0 & 0\\
	27/4 & 2 & 3\\
	35/12 & 1 & 2\\
\end{pmatrix} \quad
\text{and initial vector} \quad 
\tau \begin{pmatrix} -2 \\ 1/3 \\ 0
\end{pmatrix}  = \begin{pmatrix} 2 \\ -1 \\ -4
\end{pmatrix} \]
satisfies the original invariant 
$-11x^2+y^2-3z^2+2xy-12xz + x +z = -1$.
\end{example}

\section{Conclusion}
\label{sec:conclusion}
\subsection{Related Work}
\label{sec:related}
\subparagraph{Loop Synthesis.}
Work by Humenberger et~al.\ on loop synthesis employs an approach
based on algebraic reasoning about
linear recurrences and translating loop synthesis into an SMT (Satisfiability Modulo Theory)
solving task in non-linear arithmetic~\cite{humenberger2022journal}.
Their approach is relatively complete %
in the sense that every loop with algebraic values is captured as one of the solutions to the system of constraints.
At the same time, no method is known to decide whether such a system has a \emph{rational} solution.
In contrast, our approach gives a characterisation of quadratic invariants that have linear loops with rational values.

Another SMT-based algorithm for template-based synthesis of general polynomial programs is given in work by Goharshady et~al.~\cite{goharshady2023synthesis}. However, loops generated for an invariant $P=0$ using the latter approach necessarily have $P=0$ as an inductive invariant and %
are not guaranteed to have infinite orbits.
Recent work by Kenison et al.\ addresses the loop synthesis problem for multiple polynomial invariants, 
where each of the polynomials is a binomial of a certain type~\cite{kenison2023pdb}.
In our work, we restrict not the number of monomials in an invariant, but its degree,
and thus achieve a complete solution for a single quadratic invariant.

\subparagraph{Solving Polynomial Equations.}
As noted in
\cref{rem:SyntSolv}, 
one of the fundamental challenges towards loop synthesis arises from the study of
integer and rational solutions to polynomial equations. 
A \emph{Diophantine equation} \(F(x_1, x_2,\ldots, x_d)=0\) is a polynomial equation with rational coefficients in at least two variables.
A general decision procedure for the existence of rational solutions to a Diophantine equation (\cref{prob:dioph}) is not known. 
Over the ring of integers, this is Hilbert's 10th Problem, proven undecidable by Matiyasevich in 1970~\cite{matiyasevich1970h10}. 
Furthermore, there does not exist an algorithm that for an arbitrary
Diophantine equation, decides whether it has infinitely many integer
solutions~\cite{davis1972numsol}.

In contrast to the algorithmic unsolvability of Hilbert's 10th Problem 
and the open status of~\cref{prob:dioph}, algorithms exist that allow finding rational solutions for special classes of equations.
For instance, there exist procedures~\cite{grunewald1981solveZ,siegel1972quad,masser1998bound} 
completely solving the specialisation of the problem to quadratic
equations.
Masser introduced an approach based on the effective search bound for rational solutions~\cite{masser1998bound}.
A further improvement of this approach for $d \geq 5$ is provided in \cite{browning2008bound}.
An alternative procedure to decide whether an arbitrary quadratic equation has a rational solution is described in~\cite{grunewald1981solveZ} (see Corollary, pg.~2 therein). 
Determining the existence of integer solutions to a \emph{system}
of quadratic equations is, however, undecidable~\cite{britton1979systems}. 

\subsection{Discussion}
\label{sec:discussion}
We conclude by sketching some observations and pointing out the directions for future work.

\subparagraph{Multiple loops.}
The approach of \cref{alg:lin} can be adapted to generate multiple linear loops satisfying a given invariant.
Different solutions of the quadratic equation can be found in line~5 (\cref{alg:lin}) and subsequently used as an initial vector.
Moreover, in line~11, it is possible to pick two variables~$(x_1, x_2)$ in different ways, thus obtaining different matrices~$M$ in line~12.
Each of the matrices synthesised so is an element of the orthogonal group $\Gamma(Q')$%
\footnote{The orthogonal group~$\Gamma(Q)$ of a quadratic form~$Q$ is the group of all linear automorphisms~$M \in \GL_d(\KK)$ such that $Q(\bm{x}) = Q(M\bm{x})$ for all~$\bm{x} \in \KK^d$.} 
of the quadratic form~$Q'$.
Therefore, all possible products of these matrices also preserve the value of~$Q'$
and can be used as updates.
Other than the default matrix selected by the algorithm, 
some of these matrices alter more than two variables non-trivially,
intuitively making the synthesised loop more specific to the polynomial invariant.

\subparagraph{Number of loop variables.}
Let $P(x_1, \dots, x_d) = 0$ be a quadratic invariant in~$d$ variables.
Note that \cref{loop-for-qe} can be interpreted in terms of linear loops with variables~$x_0, x_1, \dots, x_d$.
Specifically, we can redefine the loop synthesis problem (\autoref{prob:synthesis}) by searching for linear loops with $s=d+1$ variables. %
To this end, update the procedure in~\cref{sec:qe} as follows:
if the output of the original algorithm is %
an affine loop~$\langle M, \bm{s}, \bm{t}\rangle$, 
then output the linear loop  \[\left\langle\left(\begin{array}{c|c} 1 & 0_{1,d} \\ \hline \bm{t}  & M	\end{array}\right)\!\!, \begin{pmatrix}1 \\ \bm{s} \end{pmatrix}\right\rangle.\]
Due to \cref{cor:loop-exists}, the updated procedure solves the problem of loop synthesis with one additional variable. What follows is a reinterpretation of~\cref{loop-for-qe}:
\begin{corollary}\label{lin-loop-for-qe}
	There exists an effective procedure for the following problem: 
	given a quadratic equation
	\[Q(x_1, \dots, x_d) + L(x_1, \dots, x_d) = c,\]
	decide whether there exists a non-trivial linear loop in~$d+1$ variables~$\{x_0, x_1, \dots, x_d\}$ that satisfies it.
	Furthermore, the procedure synthesises a loop, if one exists.
\end{corollary}
Increasing the number of variables in the loop template leads to the following question,
also raised in~\cite{humenberger2022journal}:%
\begin{que}
	Let $P$ be an arbitrary polynomial in~$d$ variables. Does there exist an upper bound~$N$ such that if a non-trivial linear loop satisfying~$P = 0$ exists, then there exists a non-trivial linear loop with at most~$N$ variables satisfying the same invariant?
\end{que}

\cref{lin-loop-for-qe} (together with~\cref{cor:loop-exists}) shows that, \emph{for quadratic polynomials, $N$ is at most $d+1$.}
Moreover, we show in~\cref{sec:qf} that in the class of polynomial equations $Q(\bm{x}) - c$, where $Q$ is a quadratic form, the bound~$N = d$ is tight. 
A full characterisation of quadratic equations for which linear loops with~$d$ variables exist would also be of interest.
\subparagraph{Sufficient conditions.} 
The results of~\cref{sec:qf,sec:qe} witness another class of polynomial invariants for which non-trivial linear (or affine) loops always exist.
Similar  to the setting of equations with pure difference binomials in~\cite{kenison2023pdb}, 
we can claim this for invariants~$Q(x_1, \dots, x_d)= c$ with isotropic quadratic forms~$Q$.
In particular, for every equation of the form $a_1x_1^2 + \cdots + a_dx_d^2 + c=0$ with $d \geq 4$ and $a_1, \dots, a_d, c$ not all possessing the same sign, there exists a non-trivial linear loop with~$d$ variables.
This fact is due to Meyer's Theorem on isotropy of indefinite forms~\cite{meyer1884mathematische} and \cref{cor:loop-exists}(2).
\subparagraph{Beyond quadratic.}
One future work direction concerns loop synthesis from invariants 
that are polynomial equalities of higher degrees, and, in particular, algebraic forms.
However, we are limited by the hardness of~\cref{prob:dioph}, as before.
For Diophantine equations defined with homogeneous polynomials of degree~3,
the loop synthesis is related to the study of rational points on elliptic curves, a central topic in computational number theory~\cite{silverman2015elliptic,cohen2007number}.

\newpage
\bibliography{idealbib}
\newpage
\appendix
\section{Finding isotropic vectors} \label{app:appendix-isotropy}

We consider two tasks related to quadratic forms over the rationals.
First, the problem of determining whether a quadratic form is isotropic.
Given an isotropic quadratic form, the second problem is to determine an isotropic vector for the form.
\emph{There are effective procedures for both tasks}, which we respectively call \textproc{isIsotropic} and \textproc{findIsotropic}.
We call on both of these procedures in \autoref{alg:isotropic}.

	Before we describe \textproc{solve}, we discuss certain assumptions.
Given a quadratic form $a_1x_1^2 + \dots +a_dx_d^2 -cx_{d+1}^2$, 
we can assume $a_1, \dots, a_d, c \in \ZZ$ by homogeneity.
Without loss of generality, we also assume that the coefficients $a_1, \dots, a_d, c$ are non-zero. 

A necessary condition for $a_1x_1^2 + \dots +a_dx_d^2 = c$ (\autoref{eq:pure-quadr}) to admit a non-trivial solution concerns the signs of coefficients. 
If each of the coefficients $a_1, \dots, a_d$ is positive (resp.\ negative), $c$ has to be positive (resp.\ negative); otherwise, no solutions exist.
In summary, once called, procedure \textproc{solve} first checks if the quadratic form is indefinite.
It returns ``\textsc{no solutions}'' otherwise.
Hence, we can assume that \textproc{isIsotropic} is initiated only for indefinite forms.
The subroutine \textproc{eliminateZeros} follows the proof of~\cref{rep-not-zero} 
and produces an isotropic vector where every variable is non-zero from an arbitrary isotropic vector.
\begin{algorithm} 
		\caption{returns a non-zero solution of~\cref{eq:pure-quadr} or ``\textsc{no solutions}'' as appropriate}		\label{alg:isotropic}
	\begin{algorithmic}[1]
	\Require An equation $a_1x_1^2 + \dots +a_dx_d^2 = c$, where $a_1, \dots, a_d, c \in \ZZ$.
	
	\Comment{The corresponding quadratic form $a_1x_1^2 + \dots +a_dx_d^2 -cx_{d+1}^2$ is non-degenerate.}
	\State $(\alpha_1, \alpha_2, \ldots, \alpha_{d+1}) :=$ \textsc{undefined}
	\State $(\beta_1, \beta_2, \ldots, \beta_{d+1}) :=$ \textsc{undefined}
	\Function{solve}{$a_1,\ldots, a_d, c$}
		\If{$a_1, \dots, a_d, -c$ are all of the same sign}
			\Return ``\textsc{no solutions}''
		\Else
			\If{$\textproc{isIsotropic}(a_1,\ldots, a_d, -c)= \text{``\textsc{yes}''}$}
			\Comment{isotropy test}
			\State $(\alpha_1, \alpha_2, \ldots, \alpha_{d+1}) :=$ \textproc{findIsotropic}$(a_1,\dots,a_d,-c)$. \Comment{find an iso vector}
			\State $(\beta_1, \beta_2, \ldots, \beta_{d+1}):= \textproc{eliminateZeros}(\alpha_1, \alpha_2, \ldots, \alpha_{d+1})$
			\State \Return $(\beta_1/\beta_{d+1}, \beta_2/\beta_{d+1}, \ldots, \beta_d/\beta_{d+1})$
			\Else
			\State \Return ``\textsc{no solutions}''
			\EndIf
	\EndIf
	\EndFunction
	\end{algorithmic}
\end{algorithm}

The theoretical machinery underpinning both \textproc{isIsotropic} and \textproc{findIsotropic} is beyond the scope of our work.
The purpose of this appendix is to bring together results from computational number theory and signpost relevant sources in the literature.
To this end, we categorise non-degenerate quadratic forms \(Q(x_1, \ldots, x_{d+1})\) (as in \cref{eq:iso}) according to their dimension \(d+1\).

\subparagraph*{Cases \(d\in\{0,1\}\).}
The case $d=0$ is trivial, and for $d = 1$,
\(a_1x_1^2 - cx_2^2 = 0\)
has a non-trivial solution $(\alpha_1, \alpha_2)$ if and only if $\frac{c}{a_1}$ is a square in~$\QQ$. This also yields an algorithm to compute a solution.
Finding isotropic vectors (and, to a certain extent, just testing isotropy) becomes more intricate with $d \geq 2$.

\subparagraph*{Ternary quadratic forms (\(d=2\)).} 
In \autoref{bin-solutions} we demonstrated how one can generate infinitely many non-trivial solutions from a single isotropic vector.
We now turn our attention to the problem of deciding whether a ternary quadratic form (as in \cref{eq:iso}) is isotropic and, if so, how to find an isotropic vector.

Legendre's theorem handles the problem of deciding whether a ternary quadratic form.
Recall that an integer \(q\) is a \emph{quadratic residue modulo \(n\)} if there exists an integer \(x\) for which \(x^2 \equiv q \pmod{n}\).
\begin{theorem}[Legendre] \label{thm:legendre}
	Let \(a,b,c\in\Z\) be rational integers such that \(abc\) is both non-zero and square-free.
	Then the Diophantine equation \begin{equation}\label{eq:ternary}ax^2 + by^2 + cz^2 = 0\end{equation} has non-trivial solutions if and only if 
	\begin{enumerate}
		\item \(a\), \(b\), and \(c\) do not all have the same sign, and
		\item \(-ab\), \(-bc\), and \(-ca\) are quadratic residues modulo \(c\), \(a\), and \(b\), respectively.
	\end{enumerate}
\end{theorem}
Recall that every equation of the form \eqref{eq:ternary} admits a normalised square-free form with integer coefficients.
Thus we can always assume the coefficient conditions set out in \autoref{thm:legendre}.

On the topic of solving ternary quadratics, a constructive proof of Legendre's theorem is known (see, e.g., \cite[Chapter 17.3]{ireland1990classical}). %
That is, there exist algorithms that generate a single solution to~\cref{eq:ternary}, or prove that none exist. 
We provide references to two algorithmic approaches in the literature.

Cremona and Rusin  gave a factorization-free reduction method of solving ternary quadratics~\cite{cremona2003conics}. 
Their method efficiently generates an integer solution $(x_0,y_0,z_0)$ to~\cref{eq:ternary}.

Work by Simon is built on Gauss' approach to solving ternary quadratic forms \cite{simon2005solving}.
Briefly, Gauss' approach consists of two steps: first, compute square roots modulo the form's coefficients and build another quadratic form with determinant \(-1\); and second, reduce and solve this new quadratic form.
Simon's work introduced a generalisation of the LLL–algorithm to indefinite quadratic forms, which speeds up the second step in Gauss' approach.

\subparagraph*{Quaternary quadratic forms (\(d=3\)).}
Recall  the Hasse-Minkowski  theorem, a classical global-local principle in number theory, as follows:
\begin{theorem}[Hasse--Minkowski]
	A homogeneous quadratic equation with rational coefficients represents zero in the field of rationals if and only if it represents zero in \(\R\) and in the fields \(\Q_p\), the \(p\)-adic fields, for all primes \(p\).
\end{theorem}
The standard proof of the theorem reduces the problem to considering quadratic forms as follows: \(a_1x_1^2 + a_2 x^2 + a_3 x^2 + a_4 x_4^2\) where each coefficient is a square-free integer.
Then, by way of Dirichlet's theorem on the infinitude of primes in arithmetic progressions,
we search for a prime that the two binary quadratic forms \(a_1x_1^2 + a_2 x^2\) and \(-a_3 x^2 - a_4 x_4^2\) both represent.
Thus for isotropic quarternary quadratic forms, the classical proof of the Hasse--Minkowski theorem explicitly constructs a representation of zero  (cf.~\cite{simon2005quadratic}).

As an algorithm, the aforementioned construction of a representation of zero is inefficient due to its reliance on Dirichlet's theorem.
Simon produced an efficient algorithmic construction of a representation of zero~\cite{simon2005quadratic}. 
Simon's algorithm employs a technique first seen in work by Cassels on a proof of the Hasse–Minkowski Theorem for quaternary quadratic forms that circumvents Dirichlet's theorem.

\subparagraph*{Quadratic forms with 5 and more variables (\(d\ge 4\)).}

Meyer's theorem establishes that every indefinite quadratic form in five or more variables has a non-trivial representation of zero \cite{meyer1884mathematische}.

\begin{theorem}[Meyer]
Every indefinite quadratic form in five or more variables over the field of rationals is isotropic.
\end{theorem}

Thus all that remains is to compute a representation of zero.
Fortunately Simon's algorithm for finding non-trivial isotropic vectors  for the case \(d=3\) generalises, with admittedly careful case analysis, to \(d\ge 4\).
We note that Simon's algorithm is computationally prohibitive at higher dimensions as the approach for a given quadratic form \(Q\) is equivalent to factoring the determinant of the associated matrix \(A_Q\).
Later work by Castel established an algorithm that, in polynomial time,  computes a representation of zero for a \(5\)-dimensional quadratic form that does not require the factorisation of the determinant~\cite{castel2013solving}.
Castel's algorithm generalises: in higher dimensions one restricts the form to an appropriate \(5\)-dimensional subspace.

\section{\autoref{alg:aff}} \label{app:appendix-mainalgorithm}
\begin{algorithm}[H]
	\caption{Find an affine loop satisfying a given quadratic equation} \label{alg:aff}
	\begin{algorithmic}[1]
	
	\Require an equation $Q(\bm{x}) + L(\bm{x}) = c$ in~$d$ variables (as in~\cref{eq:quadr}).	Assert $d \geq 2$.
	
	\Function{affLoop}{$Q,L,c$}
	\State $\langle M, \bm{s}, \bm{t} \rangle := \textsc{undefined}$
	\State \(\delta := \det(A_Q)\)

\LeftComment*{We use the following ternary operator shorthand for conditional expressions:
	\ConExp{\text{condition}}{\text{expressionTrue}}{\text{expressionFalse}}.}

  \vspace{1em}
  \Switch{cases of quadratic equation $(Q,L,c)$}
	\vspace{0.5em}
  	\Item{\(L(\bm{x})\equiv \bm{0}\):} \Comment{see \cref{sec:qf} and \cref{alg:lin}}
  		\State \(\langle M, \bm{s} \rangle := \textproc{linLoop}(Q,c)\)
  		 \State \algorithmicreturn \ConExp{\langle M, \bm{s} \rangle = \text{``\textsc{no loop}''}}{\text{``\textsc{no loop}''}}{\langle M, \bm{s}, \bm{0} \rangle}
	\EndItem
\vspace{1em}
  	\Item{\(\delta \neq 0\):} \Comment{see  \autoref{loop-for-qe-nondeg} for non-degenerate quadratic forms}
				\State compute $\bm{h}, \tilde{c}$ for the affine map~$f$ associated with~$Q+L=c$.
		\State $\langle M, \bm{s} \rangle := \textproc{linLoop}(Q, \tilde{c})$
  		 \State \algorithmicreturn \ConExp{\langle M, \bm{s} \rangle = \text{``\textsc{no loop}''}}{\text{``\textsc{no loop}''}}	{\langle M, \frac{1}{2\delta}\left(\bm{s}- \bm{h}\right),\frac{1}{2\delta}\left(M - I_d\right)\bm{h}\rangle}%
  	\EndItem
\vspace{1em}	
  	\Item{\(\delta = 0\):} \Comment{see \autoref{ssec:degeneratequad} for degenerate quadratic forms}
		\State compute $\tau \in \GL_d(\QQ)$ such that $\tilde{Q}(x_{k+1}, \dots, x_d) = Q(\tau\bm{x})$ is non-degenerate
		\State rewrite $L(\tau\bm{x}) =  \tilde{L}(x_{k+1}, \dots, x_d) + \lambda_1x_1 + \dots + \lambda_kx_k$
\vspace{0.5em}
		\If{$\lambda_1 = \dots = \lambda_k = 0$}
			\State compute $\delta' = \det{\tilde{Q}}$, $\bm{h}, \tilde{c}$ for the affine map~$f$ associated with~$\tilde{Q}+ \tilde{L} = c$

			\State	$\langle M, \bm{s} \rangle := \textproc{linLoop}(0x_1^2+ \dots + 0x_k^2 + \tilde{Q}(x_{k+1}, \dots, x_{d}), \tilde{c})$
			\State set $M':= \diag{2, \dots, 2, 1, \dots, 1}$ and $\bm{s'}:=\frac{1}{2\delta'}\left(\bm{s} - \begin{pmatrix}
						\bm{0} \\ \bm{h}
					\end{pmatrix}\right)$
			\State $\langle M, \bm{s},\bm{t}\rangle:=$\ConExp
				{\langle M, \bm{s} \rangle = \text{``\textsc{no loop}''}}{\text{``\textsc{no loop}''}}
				{
					\left\langle M', \bm{s'}, \bm{0} \right\rangle
				}
				
				\algstore{myalg}
\end{algorithmic}
\end{algorithm}

(continued overleaf)

\begin{algorithm}[H]                   
\begin{algorithmic} [1]                   %
\algrestore{myalg}
			\Else %

			\Comment{\small rewrite equation as $\tilde{Q}(y_1, \dots, y_r) + \tilde{L}(y_1, \dots, y_r) = c - \sum_{i=1}^k \lambda_i x_i$ with $\lambda_i \neq 0$}
			\State set $\bm{y} := \bm{\alpha}$ arbitrarily and find a solution $\bm{\beta} = (\beta_1, \dots, \beta_k)$ for $\bm{x}$
			\Switch{for degenerate quadratic equations with linear forms} %
				\Item{$k>1$:} \Comment{see \hyperlink{item:1}{Case 1} of \autoref{loop-for-qe-deg}}
					\State $\langle M, \bm{s},\bm{t}\rangle := \left\langle \begin{pmatrix}2 & 0\\-\frac{\lambda_1}{\lambda_2} & 1 \end{pmatrix} \oplus I_{d-2}, (\bm{\beta}, \bm{\alpha})^\top,\bm{0} \right\rangle$; \Break
				\EndItem
				\Item{$r>1$:} \Comment{see \hyperlink{item:2}{Case 2} of \autoref{loop-for-qe-deg}}
							\State  find $\delta' = \det{\tilde{Q}}$ and $\bm{h}, \tilde{c}$ for the map~$f$ assoc.~with $\tilde{Q}+ \tilde{L} = c-\lambda_1\beta_1$
							\State	define $\langle M, \bm{s}\rangle := \textproc{linLoop}(\tilde{Q}, \tilde{c})$ and $\bm{v} := \frac{1}{2\delta'}(M-I_r)\bm{h}$
							\State $\langle M, \bm{s},\bm{t}\rangle := \left\langle 
			\left(\begin{array}{c|c} 1 & 0_{1,r} \\ \hline \frac{1}{\beta_1}\bm{v}  & M	\end{array}\right)\!\!, 
			(\beta_1, \frac{1}{2\delta'}(\bm{s}-\bm{h}))^\top\!, \bm{0}\right\rangle$; \Break
				\EndItem
				\Item{$r = 1$ and $k = 1$:}  \Comment{see \hyperlink{item:3}{Case 3} of \autoref{loop-for-qe-deg}}
				
				\Comment{\cref{eq:QLA} has the form $ax^2+bx = c - dy$ with $d\neq 0$}
					\State $\left\langle M, \bm{s},\bm{t}\right\rangle := \left\langle \begin{pmatrix}
							2 & 0 \\
							2\frac{b}{d} & 4
						\end{pmatrix},
			 (1, (c-a-b)/{d})^\top,
			 (0, -3c/{d})^\top 
			\right\rangle$
				\EndItem
			\EndSwitch
		\EndIf
\State \algorithmicreturn
\ConExp{\langle M, \bm{s},\bm{t} \rangle = \text{``\textsc{no loop}''}}{\text{``\textsc{no loop}''}}	{\langle \tau M \tau^{-1}, \tau\bm{s}, \tau\bm{t}\rangle}
  	\EndItem 
  \EndSwitch
  \EndFunction
\end{algorithmic}
\end{algorithm}

\end{document}